\documentclass[a4paper,10pt]{article}

\usepackage[utf8]{inputenc}
\usepackage[T1]{fontenc}
\usepackage[english]{babel}
\usepackage{graphicx}
\usepackage{makeidx}
\usepackage{amssymb,amsmath,amsthm}
\usepackage{amsmath, amsfonts, amsthm}
\usepackage{mathtools}
\usepackage{bm}
\usepackage{subfig} 
\usepackage{upgreek}
\usepackage{mathrsfs}
\graphicspath{{immagini/}}
\usepackage{caption}
\usepackage{MnSymbol}
\usepackage{xcolor}
\usepackage{authblk}
\newtheorem{thm}{Theorem}
 \theoremstyle{definition}
\newtheorem{definizione}{Definition}
 
\DeclarePairedDelimiter{\norma}{\lVert}{\rVert}

\renewcommand{\emptyset}{\varnothing}

\definecolor{myred}{rgb}{0.9,0,0}
\newcommand{\rosso}[1]{\textcolor{myred}{#1}}
\newcommand{\comm}[1]{\rosso{({\em #1})}}
\newcommand{\wuno}{{{{\bf z}}_{1}}}
\newcommand{\wdue}{{{{\bf z}}_{2}}}
\newcommand{\wi}{{{{\bf z}}_{i}}}
\newcommand{\vett}[1]{\boldsymbol{#1}}
\newcommand{\N}{\mathbb{N}}
\newcommand{\Z}{\mathbb{Z}}

\newcommand{\R}{\mathbb{R}}

\newcommand{\E}{\mathbb{E}}
\newcommand{\eps}{\epsilon}
\newcommand{\EKP}{{\rm E}_{\rm KP}}
\newcommand{\Eloop}{{\rm E}_{\rm loop}}
\newcommand{\hau}{\mathscr{H}}
\newcommand{\U}{U_{M,n_i,\eta_i}}

\title{Soap film spanning an elastic link}

\author[1]{Giulia Bevilacqua}
\author[2]{Luca Lussardi}
\author[3]{Alfredo Marzocchi}
\affil[1]{\small{MOX - Dipartimento di Matematica, Politecnico di Milano, Italy}}
\affil[2]{\small{DISMA- Dipartimento di Scienze Matematiche ``Giuseppe Luigi Lagrange'', Politecnico di Torino, Italy}}
\affil[3]{\small{Dipartimento di Matematica e Fisica ``Niccolò Tartaglia'', Università Cattolica del Sacro Cuore}}

\date{}

\begin{document}
\maketitle
\section{Abstract}

We study the equilibrium problem of a system consisting by several Kirchhoff rods linked in an arbitrary way and tied by a soap film, using techniques of the Calculus of Variations. We prove the existence of a solution with minimum energy, which may be quite irregular, and perform experiments confirming the kind of surface predicted by the model.

\section{Introduction}
In this article we find the solution of the Kirchhoff-Plateau problem which is physically motivated by soap-films that span flexible loops.
This approach is a generalization of the so-called {\em Plateau problem}, a centuries-old mathematical problem investigated by the Belgian physicist Joseph Plateau \cite{plateau}. In contrast with Plateau problem, in which a soap film spans a fixed frame, the Kirchhoff-Plateau problem concerns the equilibrium shapes of a system in which a \emph{flexible} filament has the form of a closed loop spanned by a liquid film. We model our filament as a Kirchhoff rod: it has to be thin enough, unshearable and inextensible and it can sustain bending of its midline and twisting of its cross-sections (see for instance Antman \cite{antman} Ch.~8). In this way the problem becomes ``elasto-variational''.

This kind of problem has been investigated in \cite{ggg} by Fried {\em et al.} where they consider only a filament, while our aim is to study more complex configurations of the bounding loop, like a finite number of them linked in an arbitrary way. For the sake of simplicity we will consider throughout the paper two thin elastic three-dimensional closed rods, i.e.\ two loops, linked in a simple but nontrivial way: we impose that the midline of each rod has to have linking number equal to one with the other one: this implies that they form what is called a link (see Fig. \ref{sistema}), but the case of a number of $N$ loops possibly non isotopic to a torus and arbitrarily linked can be easily treated with the same technique and minor changes. In this way, the major difference with respect to \cite{ggg} is the fact that the second loop doesn't have a fixed position in space, while the first has a prescribed frame at a point. 
\begin{figure} [h]
\centering
\includegraphics[width=0.5\textwidth]{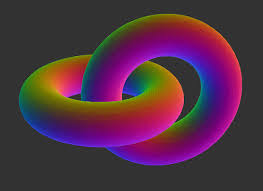}
\caption{Geometry of the problem}
\label{sistema}
\end{figure}
To take into account all of these requests we have to impose some physically motivated constraints, such as local and global non-interpenetration of matter (though allowing for points on the surface of the bounding loop to come into contact), and other already introduced by Schuricht \cite{S}, adding the necessary specifications in considering a link and not a single loop.

As for the energy functional of the system, we consider three contributions: the elastic and the potential energy for the link and the surface tension energy of the film. Precisely, we do not take into account the energy associated with the liquid/solid interface since it is less then the one between the liquid/air interface. However, for the future it could be a parameter to add to the problem in order to give a more physical description also for the spanning surface.

The most delicate point in our elasto-variational problem is a good definition of spanning surface, since we do not prescribe {\em a priori} the region where
 the soap film touches the surface of the bounding loop.\\
To overcome this problem we use the definition of {\em spanning surface} introduced by Harrison \cite{Ha} based on the concept of {\em linking number}, which is a numerical invariant well-known in topology. Even if this approach describes all soap-film solutions (\cite{Ha}, \cite{H}), it needs much strong regularity. Therefore, like in \cite{ggg}, we use a recent and powerful reformulation by DeLellis {\em et al.}\ \cite{delellis}, who formulate the Plateau problem in a particular notion of ``bounding'' and make use of Hausdorff topology for the convergence of surfaces. In this way the minimum of the problem is defined as the support of a Radon measure, \cite{david}, \cite{depauw}, \cite{reifenberg}.
Hence, this approach has the advantage of considering also non-rectifiable or not fixed boundaries but it is not easy to apply since the minimization of Hausdorff measures on classes of compact sets could cause lack of lower semicontinuity (it depends on the notion of convergence adopted), which is fundamental for the direct method of the Calculus of Variations. To avoid these difficulties we combine Preiss rectifiability theorem \cite{preiss} for Radon measure in combination with some variational arguments, such as the introduction of cone and cup competitors.

Finally, we tried to reproduce the physical counterpart problem in the laboratory in order to ``see'' experimentally the solution. The results are interesting because, as we expected, there could be a balance between the film energy and the weight in cases when the loops are very light.

\section{Formulation of the problem}



We consider two continuous bodies whose reference, or material, configurations are two right cylinders of lengths $L_1,L_2$. 
The arc-length parameter $s$ of the axis of each cylinder identifies a material (cross) section $\mathcal A(s)$, which consists of all points on a plane perpendicular to the axis at $s$ belonging to a simply connected and compact subset of the plane.
Like in (\cite{antman}, Ch.~8) we describe each rod by three vector-valued functions $[0,L_i] \rightarrow\R^3$ given by $s\mapsto (\vett r_i(s),\vett{u}_i(s),\vett{v}_i(s))$ ($i=1,2$).

Now we fix a point $O$ in the euclidean space $\E^3$ and describe the position in space of each point of the $i$th rod. Setting $G_i(s)-O=\vett r_i(s)$ (the so-called {\em midline}), where $G_i(s)$ is the center of mass of the cross-sections and considering $\vett u_i$ and $\vett v_i$ as applied vectors in $G_i(s)$, a generic point $P_i$ of the rod in space is given by the knowledge of the vector
\begin{equation}
\label{pscelta}
\vett p_i(s, \zeta_1,\zeta_2) = P_i-O= \vett r_i(s) + \zeta_1\vett u_i(s) + \zeta_1\vett v_i(s),
\end{equation}
where $(s,\zeta_1,\zeta_2)\in\Omega_i := \{ (s,\zeta_1,\zeta_2) | \, s \in [0,L_i], \, (\zeta_1,\zeta_2) \in \mathscr{A}_i(s)\}$. Hence, $\Omega_i$ is the closure of an open set in $\R^3$. Moreover, $\zeta_1$ and $\zeta_2$ are not completely free: we require that our body is {\em ``longer than broad''}, so there exists an $R>0$, the maximum thickness, which has to be small compared to the length $L_i$, such that $|\zeta_1| < R$ and $|\zeta_2| < R$ for any $(s, \zeta_1,\zeta_2)$.\\
Moreover, we also assume that the rod is {\em unshearable}, i.e.\ the cross section at any point of the midline remains in the plane orthogonal to the midline at that point, so that $\vett u$ and $\vett v$ are orthogonal to the midline, and that this line is {\em inextensible}. Hence, by these assumptions, we can choose the Kirchhoff rod as a model for the first rod, which is a special case of a Cosserat rod.\\

Given the function $\mathcal A_i(s)$, the position of the midline of each rod is then completely determined by three scalar parameters with a physical meaning: $k'_i$ and $k''_i$ are the {\em flexural densities} and $\omega$ the {\em twist density}. The vectors $\vett r_i,\vett u_i,\vett v_i$ satisfy the system of Ordinary Differential Equations
\begin{equation}
\label{s}
\left \{
\begin{aligned}
\dot{\vett r}_i(s) &= \vett w_i(s), \\
\dot{\vett u}_i(s) &= -\omega(s)\vett w_i(s) - k'_i(s)\vett v_i,\\
\dot{\vett v}_i(s) &= k'_i(s)\vett u_i(s) + k''_i(s)\vett w_i(s);
\end{aligned}
\right.
\end{equation}
where $i=1,2$ and $\vett w=\vett u\times\vett v$ is tangent to the midline.

Up to now, the two rods are defined in the same way. We now suppose that the first one is ``clamped'' by assigning an initial value to its system, i.e. 
\begin{equation}
\label{eq:terza}
(\vett r_1(0),\vett u_1(0),\vett v_1(0))=(\hat{\vett r}_1,\hat{\vett u}_1,\hat{\vett v}_1).
\end{equation}
Since clearly 
$$\dot{\vett w}_1(s)=-\omega_1(s)\vett u_1(s)-k''_1(s)\vett v_1(s)$$
the triple $(\vett u_1,\vett v_1,\vett w_1)$ satisfies a non-autonomous linear system and therefore, then by classical results \cite{hartman}, if the densites $k'_1,k''_1$ and $\omega$ belong to $L^p ([0,L_1];\R)$ for some $p \in (1,\infty)$, then the initial-value problem has a unique solution, with $\vett r_1 \in W^{2,p}([0,L_1];\R^3)$ and $\vett u_1,\vett v_1 \in W^{1,p}([0,L_1];\R^3)$.

It is easy to verify that if $(\hat{\vett u}_1,\hat{\vett v}_1,\hat{\vett w}_1)$ is orthonormal, so is $(\vett u_1(s),\vett v_1(s),\vett w_1(s))$ for every $s \in [0,L_1]$. For every $(\hat{\vett u}_1,\hat{\vett v}_1,\hat{\vett w}_1)\in(\R^3)^3$ we then set
$$\wuno = (k'_1,k''_1,\omega_1)\in{V}_1 := L^p ([0,{L}_1];\R^3).$$
 
As for the second rod, since we do not know {\em a priori} its position in space, we need some information also on the orientation of one of its orthonormal frames. Therefore we seek a solution of the form
$$
\wdue = (k'_2,k''_2,\omega_2,\hat{\vett r}_2,\hat{\vett u}_2,\hat{\vett v}_2,\hat{\vett w}_2)\in{V}_2 := L^p ([0,{L}_2];\R^3) \times \R^3 \times \R^3 \times \R^3 \times \R^3
$$
where $\hat{\vett u}_2,\hat{\vett v}_2,\hat{\vett w}_2$ are orthonormal and $\hat{\vett r}_2$ gives their application point. 

Now the system \eqref{s}${}_{2,3}$ and \eqref{eq:terza}, together with the knowledge of $\hat{\vett r}_2$, fully fixes the position in space of the second midline.

Since we want to deal with closed loops, we have to restrict to a suitable subclass of descriptors by imposing topological constraints. Obviously we impose the closure of the midlines, {\em i.e.} 
\begin{equation}
\label{midline}
\vett r_i(0)=\vett r_i(L_i)\qquad (i=1,2)
\end{equation}
and, since we do not want interpenetration, we need to have also continuity of the tangent vectors, so that for $i=1,2$ 
\begin{equation}
\label{directors}
\vett w_i(0)=\vett w_i(L_i).
\end{equation}

The simple determination of the midline, however, does not completely fix the shape of the loops if they are three-dimensional. Indeed, the same midline may correspond to different bodies if the cross-sections ${\mathcal A}_i(s)$ are rotated around the midline before being glued, and the final rotation angle depends on the shape of the cross-section. On the other hand, since they are undeformable, the information to be encoded reduces to fixing a point in every section. First of all we recall the notion of {\em isotopy}, which will be useful also later on.

\begin{definizione}
Let ${\eta}_{i} : [a,b] \to \R^3$, with $i=1,2$, be two continuous curves with ${\eta}_{i}(a)= {\eta}_{i}(b)$. ${\eta}_1$ and ${\eta}_2$ are said to be {\em isotopic}, ${\eta}_1 \simeq {\eta}_2$, if there are open neighborhoods $N_1$ of ${\eta}_1([a,b])$, $N_2$ of ${\eta}_2([a,b])$ and a continuous mapping $\Phi : N_1 \times [0,1] \mapsto \R^3$ such that $\Phi(N_1,\tau)$ is homeomorphic to $N_1$ for all $\tau \in [0,1]$, $\Phi (\cdot, 0)$ is the identity, $\Phi (N_1,1)= N_2$ and $\Phi({\eta}_1([a,b]),1) = {\eta}_2([a,b])$.
\end{definizione}

The isotopy class is then stable with respect to diffeomorphism and define also the {\em knot type}. Another very useful notion is the linking number.

\begin{definizione}
\label{linking}
Let $\eta_1,\eta_2$ be two absolutely continuous disjoint closed curves in $\E^3$. The number
$$
L(\eta_1,\eta_2)=\frac{1}{4\pi}\int_a^b\int_a^b \frac{\eta_1(s)-\eta_2(t)}{|\eta_1(s)-\eta_2(t)|^3}\cdot(\eta_{1}'(s)\times\eta_{2}'(t))\,ds\,dt
$$
is called the {\em linking number} between $\eta_1$ and $\eta_2$.
\end{definizione}

It is well-known \cite{munkres} that $L$ is always an integer and that is invariant in the isotopy class of the two curves.

To encode a possible rotation of the cross-sections, we then proceed as follows for each of the two rods. Since the thickness is nonzero, we can consider a curve ``near'' the midline ${\vett r}_i$, which could be not a closed one since the endpoints may be different. Joining them without interesecting the midline, we obtain a closed curve which has a certain linking number with the midline. Of course, every possible midline has to preserve this number, so we will impose this constraint on the midline. At this point, once we know the midline, the position of the nearby curve is fixed and so is its every cross-section, thus completely defining\footnote{Up to a set of $L^1$-zero measure which is irrelevant.} the shape of the loops, which we will indicate by $\Lambda[{\bf z}]$, see Fig \ref{loop}.

Finally, we want to impose that the two loops form a link. We then come back to the two midlines, and we suppose that they are linked with a given linking number $L_{12}\in\Z$. As they are closed sets, they admit disjoint neigbourhoods, which we can suppose tubular without loss of generality (\cite{mukherjee} pp 199-223). By a further shrinking to the diameter of ${\mathcal A}(s)$ we have that both rods are disjoint and linked one each other with the given linking number.

At this point, the shape of the two solids is assigned once we know $\wuno,\wdue$, but we still have to avoid local and global interpen\-etration, which is clearly unphysical. To this end, we first introduce the elastic and potential energy stored in the loops.

The elastic energy is supposed to be of the classical form (see for instance \cite{dacorogna}, Ch.~2)
\begin{equation}
\label{esh2}
{\rm E}_{\rm el_{\tiny i}}[\wi] := \int_{0}^{L_i} f_i(\wi(s),s)\, ds
\end{equation}
where $f_i(\cdot, s)$ are continuous and convex for any $s \in [0,L_i]$ and $f_i(a,\cdot)$ is measurable for any $a \in \R^3$. Since we are going to apply the Direct Method of the Calculus of Variations, we suppose that there exist positive constants $C_i, D_i$ such that
\begin{equation}
f_i(a,s)\geq C_i|a|^p+D_i \qquad \forall (a,s) \in \R^3 \times [0,L_i].\\
\label{eq:effe_i}
\end{equation}
In view of this, the total elastic energy
$$
{\rm E}_{\rm el}[{\bf z}] = {\rm E}_{\rm el_{\tiny 1}}[\wuno] +  {\rm E}_{\rm el_{\tiny 2}}[\wdue] := \int_{I} f({\bf z}(\xi),\xi)\, d\xi,
$$
where $I=[0,L_1]\times [0,L_2]$, ${\bf z} = (\wuno,\wdue)$ and $\xi$ is a vector variable, is easily seen to be coercive on $V:= V_1 \times V_2$.

As for the potential energy of the weight, it is given for each loop by
$$
{\rm E}_{{\rm g}_{\tiny i}}[\wi]=-\int_{0}^{L_i}\rho_i(s)\,\vett g\cdot (G_i(s)-O)\, ds
$$
where $\rho_i>0$ stand for the mass of each section of the rod and $\vett g$ denotes the acceleration of gravity.\\

It is worth insisting on the fact that the weight plays a different role in the two rods: in the first it acts essentially deforming only the midline, while in the second it influences the global positioning of the rod, and could draw it away without appropriate conditions of non intersection, that we will introduce below. 

We also set 
$$\Eloop[{\bf z}]={\rm E}_{\rm el_{\tiny 1}}[{\wuno}] + {\rm E}_{{\rm g}_{\tiny 1}}[{\wuno}]+{\rm E}_{\rm el_{\tiny 2}}[{\wdue}] + {\rm E}_{{\rm g}_{\tiny 2}}[{\wdue}].$$

We now need to set sufficient conditions for the local and global non-interpen\-etration of our configuration. 

As for the first, it is well-known (\cite{antman}, Theorem 6.2, p.276) that for Kirchhoff rods the condition is equivalent to the existence of two convex, homogeneous functions $g_i(\xi_1,\xi_2,s)$ such that $g(0,0,s)=0$ and
\begin{equation}
\begin{aligned}
\label{preserve2}
g_i(k'_i(s),k''_i(s),s)<1  \quad &\hbox{for a.e. } s \in [0,L_i],\qquad (i=1,2).\\
\end{aligned}
\end{equation}

However, this will not define a weakly closed set in the space of solutions, due to the strict inequality. Therefore, we will require the weaker condition
\begin{equation}
\begin{aligned}
\label{preserve3}
g_i(k'_i(s),k''_i(s),s)\leq 1  \quad &\hbox{for a.e. } s \in [0,L_i],\qquad (i=1,2).\\
\end{aligned}
\end{equation}
even if this could let some point to infinite compression, and to prevent this we impose the natural growth condition on the elastic energy as
\begin{equation}
\label{fi}
\begin{aligned}
f_i({\wi}(s),s) \to +\infty \qquad \hbox{as} \qquad g_i(k'_i(s),k''_i(s),s) \to 1, \, (i=1,2)\\
\end{aligned}
\end{equation}
i.e. the elastic energy approaches infinity under complete compression (remember thay $f_i$ may depend on $g_i$). By this assumption we have that the equality in \eqref{preserve2} can occur only on a set of measure zero for configurations with finite energy.

At this point it is not difficult to prove the

\begin{thm}
\label{thm1}
Let ${\bf z}=(\wuno, \wdue) \in V = {V}_1 \times {V}_2$ satifies \eqref{preserve3}, $f_i$ with $i=1,2$ satisfies \eqref{fi} and ${\rm E}_{\rm el}({\bf z}) < +\infty$.\\
Then the mapping $(s, \zeta_1,\zeta_2) \mapsto \vett{p}[{\bf z}](s,\zeta_1,\zeta_2) = (\vett{p}_1,\vett{p}_2)[\wuno, \wdue](s,\zeta_1,\zeta_2)$ is locally injective on ${\rm int} \, \Omega$. Moreover, this mapping is open on ${\rm int} \, \Omega$.
\end{thm}
\begin{proof}
The proof can be made easily by following the proof presented in \cite{ggg} taking into account the fact that we can study the two rods separately since it is sufficient to reduce the proof in an open neighbourhood well-cointained in each rod. 
\end{proof}

As for the global injectivity, we must distinguish each loop and their union. First of all, Ciarlet and Ne\v{c}as \cite{ciarlet} proved that if this condition holds \eqref{gi12}
\begin{equation}
\label{gi12}
\int_{\Omega_i} {\rm det} \, \frac{\partial \vett{p}_i (s,\zeta_1,\zeta_2)}{\partial(s, \zeta_1, \zeta_2)} \, \mbox{d} (s,\zeta_1, \zeta_2) \leq \mathscr{L}^3(\vett{p}_i[{\bf z}_i](\Omega_i)),
\end{equation}
the global injectivity is true. Moreover, in our case it can be rewritten as
\begin{equation}
\label{gi1}
\int_{\Omega_i} (1- \zeta_1k'_i(s) - \zeta_2k''_i(s)) \, \mbox{d}(s,\zeta_1,\zeta_2) \leq \mathscr{L}^3(\vett{p}_i[\wi](\Omega_i)). %
\end{equation}
Hence, assuming \eqref{gi1} true, one has the global injectivity of the functions $\vett{p}_i$ on each rod. Roughly speaking, this condition guarantees that parts of the rod which are far away from each other in the reference configuration, cannot penetrate each other after large deformations.\\
We will then suppose \eqref{gi1} for the non-interpenetration of each rod. At this point, for the union of the two, we notice that the midlines (which are closed sets) have to be disjoint and therefore there exist $R>0$ such that the maximum diameter of the sections is less than $R$ it holds 
\begin{equation}
\label{noint}
\forall \, {\bf z} \in V \quad \vett{p}_1[\wuno](\Omega_1) \cap \vett{p}_2[\wdue](\Omega_2) = \emptyset.
\end{equation}

We will then suppose the sections so small that \eqref{noint} is verified.

Now we can prove the
\begin{thm}
\label{2th}
Let ${\bf z}$ be an element of $V = {V}_1 \times {V}_2$ such that $\Eloop[\bf w] < +\infty$ and $f_i$ with $i=1,2$ satisify \eqref{fi}. Suppose that $\vett{p}_1[{\bf z}](\Omega_1) \cap \vett{p}_2[{\bf z}](\Omega_2) = \emptyset$ and $\wi$ satisfies \eqref{preserve3} and \eqref{gi1}.\\
Then the mapping $(s,\zeta_1,\zeta_2) \mapsto \vett{p}[{\bf z}](s,\zeta_1,\zeta_2)$ is globally injective on ${\rm int} \, \Omega$.
\end{thm}
\begin{proof}
By \eqref{noint} it suffices to show the global injectivity of $\vett{p}_1$ and repeat the arguments for $\vett p_2$.
Let us fix a configuration $\wuno \in {V}_1$, by Theorem \ref{thm1} there is a set $\mathcal{I}_0$ of measure zero such that
$$
\limsup_{(\tilde{s},\tilde{\zeta_1},\tilde{\zeta_2}) \to (s, \zeta_1, \zeta_2)} \frac{\| \vett{p}_1(\tilde{s},\tilde{\zeta_1},\tilde{\zeta_2}) - \vett{p}_1(s, \zeta_1, \zeta_2) \|}{\| (\tilde{s},\tilde{\zeta_1},\tilde{\zeta_2}) - (s, \zeta_1, \zeta_2)\|} < \infty \quad \forall (s, \zeta_1, \zeta_2) \in \Omega'_1,
$$
where $\Omega'_1$ is defined as $\Omega_1 \setminus \Omega_1(\mathcal{I}_0)$ and
$$
\Omega_1(\mathcal{I}_0) = \{ (s, \zeta_1, \zeta_2) \in \Omega_1 : s \in \mathcal{I}_0\}.
$$
Obviously\footnote{The complete proof of this statement is in \cite{S}. However, we can give a simple and empirical idea of the proof: since $\mathcal{I}_0$ has measure zero, the cross-sections such that their arc-length parameter $s$ belongs to $\mathcal{I}_0$ are the elements of the set $\Omega_1(\mathcal{I}_0)$. Hence, $\mathscr{L}^3(\Omega_1(\mathcal{I}_0))=0$. But now, since $\vett{p}_1$  is a regular function, we can think that {\em sets of measure zero are mapped into sets of measure zero}.} $\mathscr{L}^3(\vett{p}_1[\wuno](\Omega_1(\mathcal{I}_0)))$ is equal to zero. By the coarea formula\footnote{Let consider two locally Lipschitz continuous functions $f:\R^n\to\R$ and $g:\R^n\to\R^m$ with $m \leq n$. Then
$$
\int_{\R^n}f(x) \sqrt{{\rm det} \, [{\rm d}g(x)({\rm d}g(x))^{T}]} \, {\rm d}\mathscr{L}^n(x) = \int_{\R^m} \left( \int_{g^{-1}(y)} f(x) \, {\rm d}\hau^{n-m}(x) \right) \rm{d}\mathscr{L}^m(y).
$$
If $m=n$, the quantity $\hau^{0} = {\rm card}$, i.e. the function which counts the elements of a set \cite{federer}.} (\cite{federer}, pp 243-244), we have
\begin{equation}
\label{federer}
\int_{\Omega'_1}(1-\zeta_1{k}'_1 - \zeta_2{k}''_1) \, {\rm d} (s,\zeta_1,\zeta_2) = \int_{\vett{p}_1(\Omega'_1)} {\rm card} \{ \vett{p}^{-1}_1(\vett{q})\} \, {\rm d} \vett{q},
\end{equation}
where $\vett{p}^{-1}_1$ is the inverse of the mapping $\vett{p}_1$. Therefore, using \eqref{gi1} and \eqref{federer}, it yields
\begin{equation*}
\begin{split}
\mathscr{L}^3(\vett{p}_1[\wuno](\Omega_1)) = \int_{\vett{p}_1[\wuno]\Omega_1)} {\rm d} \vett{q} = \int_{\vett{p}_1[\wuno](\Omega_1(\mathcal{I}_0))} {\rm d} \vett{q} \leq \\
\int_{\vett{p}_1[\wuno](\Omega_1(\mathcal{I}_0))} {\rm card} \{ \vett{p}^{-1}_1(\vett{q})\} \, {\rm d} \vett{q} = \int_{\Omega'_1}(1-\zeta_1{k}'_1 - \zeta_2{k}''_1) \, {\rm d} (s,\zeta_1,\zeta_2) = \\
\int_{\Omega_1}(1-\zeta_1{k}'_1 - \zeta_2{k}''_1) \, {\rm d} (s,\zeta_1,\zeta_2) \leq \mathscr{L}^3(\vett{p}_1[\wuno](\Omega_1)).
\end{split}
\end{equation*}
Hence,
\begin{equation}
\label{card}
{\rm card} \{ \vett{p}^{-1}_1(\vett{q})\} = 1 \qquad \hbox{ for almost all } \vett{q} \in \vett{p}_1[\wuno](\Omega_1),
\end{equation}
which combined with Theorem \ref{thm1} ensure the injectivity of $\vett{p}_1$ on ${\rm int} \, \Omega_1$ and then the global injectivity of $\vett{p}$ on ${\rm int}\, \Omega$. 
\end{proof}

Finally, the energy stored in a film that will deform the link is defined as 
\begin{equation}
\label{measure}
{\rm E}_{{\rm film}}(S) = 2\sigma \hau^{2}(S),
\end{equation}
where $\hau^{d}$ represents the $d$-dimensional Hausdorff measure.
When a soap film is in stable equilibrium, as in eq. \eqref{measure}, any small change in its area, $S$, will produce a corresponding change in its energy ${\rm E}$, providing $\sigma$ remains constant.
As ${\rm E}_{{\rm film}}$ is minimized when the film is in stable equilibrium, $S$ will be minimized. Precisely, in \eqref{measure}, we do not consider what happens between the film and the bounding loop, i.e.\ the energy associated with the liquid/solid interface.  

Anyway, we still cannot provide the final expression for the energy since we have not yet specified how the film is attached to each loop.
Since in our case we have a boundary with non vanishing thickness, to formulate the idea of a solution we have to give a good definition of the terms {\em surface}, {\em area} and {\em contact}, which we will call {\em span}. We need also a precise mathematical formulation of the conditions which explain how the liquid film spans the bounding loop without detaching from it, after which we will end up with the final expression of the functional to be minimized. We begin with some recalls of topology.

\begin{definizione}
\label{link}
Let $H=\bigcup_{j\in J}H_j$ be a closed compact $3$-dimensional submanifold of $\E^3$ consisting of connected components $H_j$. We say that a circle $\gamma$ embedded in $\E^3 \setminus H$ is a {\em simple link} of $H$ if there exists $i \in J$ such that the linking numbers $L(\gamma, H_j)$ verify 
$$|L(\gamma, H_i)|=1,\qquad L(\gamma, H_j)=0\quad j\neq i.$$
\end{definizione}

Clearly, a simple link ``winds around'' only one component of $H$ (see figure \ref{loop}). Precisely, the definition of the linking number between a closed subset and a curve is exactly the one given before (Definition \ref{linking}) by considering the compactification of the $\E^3$ (for more details see \cite{rolfsen}, pp.132-136).

\begin{definizione}
We say that a compact subset $K \subseteq \E^3$ {\em spans} $H$ if every simple link of $H$ intersects $K$.
\end{definizione}

This idea is crucial: we need spanning sets (in simple cases, surfaces) crossing every simple link: in this way it is impossibile for $K$ to be ``detached'' from $H$, or having ``holes'' which are not occupied by other components of $H$ (see figure \ref{loop}).
However, in our case we need a still more general definition, because in our problem $H$ is not given a priori since $H = \Lambda[{\bf z}]$, i.e.\ it depends on the considered configuration.

Now let $H$ be an arbitrary closed subset of $\E^3$ and consider the family
$$
C_{H} = \{ \gamma : {\mathbb{S}}^1 \to \E^3 \setminus H : \gamma \hbox{ is a smooth embedding of } {\mathbb{S}}^1 {\hbox{ into }} \E^3 \}.
$$
A set $C \subseteq C_{H}$ is said to be {\em closed by homotopy} (with respect to $H$) if it contains all elements belonging to the same homotopy class. 

\begin{definizione}
\label{surface}
Given $C \subseteq C_H$ closed by homotopy, we say that a relatively closed subset $K \subset \E^3 \setminus H$ is a {\em C-spanning} set of $H$ if
$$K \cap \gamma \neq \emptyset \quad \forall \gamma \in C.$$
We denote by $F(H,C)$ the family of all $C$-spanning sets of $H$.
\end{definizione}
Notice that the set spanned by the surface, can be any closed set in $\E^3$, so we can consider $H = \Lambda[{\bf z}]$ with finite cross-section, as in our case and not only a line as in the Plateau's problem. 
Nevertheless, the spanning surface depends only on the choice of the homotopy class and not to the configuration ${\bf z}$.
Hence, we can define
\begin{definizione}
\label{dspanning}
We call a set $D_{\Lambda[{\bf z}]} \subseteq C_{\Lambda[{\bf z}]}$ a {\em $D_{\Lambda[{\bf z}]}$-spanning} set of $\Lambda[{\bf z}]$ if it contains all the smooth embeddings $\gamma$ which are not homotopic to a constant and which have linking number one with both rods. For the sake of brevity, we will write $D$ in place of $D_{\Lambda[{\bf z}]}$.
\end{definizione}
Finally, we denote $F(\Lambda[{\bf z}], D)$ the family of $D$-spanning sets of $\Lambda[{\bf z}]$ with linking number one with both  components (see Fig \ref{loop}).
\begin{figure}[h!]
\centering
\includegraphics[width=0.65\textwidth]{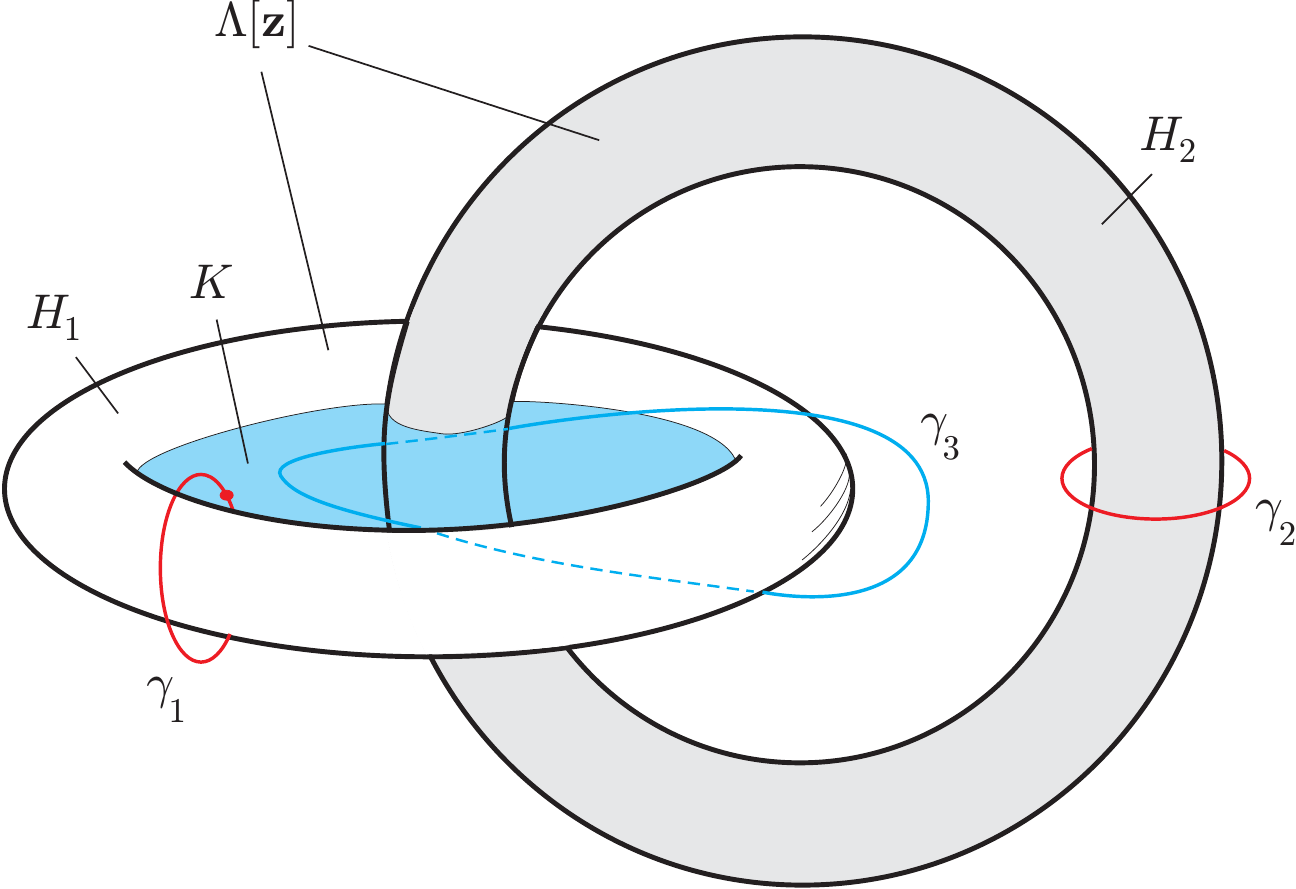}
\caption{$\gamma_i$ $(i=1,2)$ are simple links for $H_i$ while $\gamma_3\in F(\Lambda[\bf w])$. Even if $K$ is not $D$-spanning for the whole system, notice how $\gamma_1\cap K\neq\emptyset$.}
\label{loop}
\end{figure}


We are now in position to set the energy functional for our problem. We set
\begin{equation}
\label{EKP}
\EKP [{\bf z}]:= \Eloop[{\bf z}] + \inf \{ {\rm E}_{\rm film}(S) : S \, \hbox{is a } D\hbox{-spanning set of } \Lambda[{{\bf z}}]\},
\end{equation}
where ${\bf z} \in V$ and verifies all the above-mentioned constraints. Precisely, the $\inf$ in the equation \eqref{EKP} is necessary since we want to eliminate the dependence on the spanning surface $S$ and writing everything in the terms of the configuration ${\bf z}$ only.

At this point a first important result holds
\begin{thm}
\label{Umneta}
Let two circumferences $\eta_i : [0, L_i] \to \E^3$ and $M \in \R$ and $n_1, \, n_2 \in \Z$ three constants be given. Then, the set
\begin{equation}
\begin{aligned}
\U:=\{{\bf z} = (\wuno, \wdue, \hat{{\vett r}}_2, \hat{{\vett u}}_2,\hat{{\vett v}}_2) &\in V = V_1 \times V_2 : \Eloop[{\bf z}] < M;\\
&\eqref{midline}, \eqref{directors}, \eqref{gi1} \mbox{ and } \eqref{noint} \mbox{ hold }L({\bf \wi}) = n_i; \\
& L_{12} = 1 \mbox{ and } ({\vett r}_1[\wuno], {\vett r}_2[\wdue]) \simeq (\eta_1, \eta_2)\}
\end{aligned}
\end{equation}
is weakly closed in $V$.\\
\end{thm}
\begin{proof}
If $\U=\emptyset$ the thesis is obviously true.\\
If $\U \neq \emptyset$, it is an extension to Schuricht's theorems $3.9$, $4.5$ and $4.6$ \cite{S}, by remembering that
$$
\bm{p}=({\vett p}_1,{\vett p}_2) \qquad {\bf z}=(\wuno,\wdue).
$$
Moreover, since $\bm{p}_i$ are two open maps and $\Omega_i$ are the closure of two open sets in $\E^3$, condition \eqref{noint} yields the intersection of two closed sets, which is obviously a closed set and concludes the proof.
\end{proof}

\section{Main results}

Now we want to prove the existence of a solution to the Kirchhoff-Plateau problem, {\em i.e.} the existence of a minimizer of $\EKP$ given by \eqref{EKP} in the class $\U$. As a first step we find a minimizer of each its two terms. Obviously, we cannot say that the solution of our problem will be the sum of them, but this will be useful to prove the main result.

\subsection{Energy minimizer for the bounding loop}

For the first term, the functional $\Eloop$, we use a quite straightforward application of the direct method of the Calculus of variations. Recall that its expression is
$$
\begin{aligned}
\Eloop[{\bf z}] : V &\to \R \cup \{ +\infty\}\\
{\bf z} &\mapsto \Eloop[{\bf z}] = {\rm E}_{\rm el}[{\bf z}] + {\rm E}_{\rm g}[{\bf z}] = \int_{I} f({{\bf z}}(\xi),\xi) \, d\xi + {\rm E}_{\rm g}[{\bf z}],
\end{aligned}
$$
In order to verify if we can apply this method to $\Eloop$, we follow the following steps.

First, we need to show that $\Eloop$ is bounded from below and proper, {\em i.e.} $\Eloop\neq +\infty$.
The second condition is by definition. For the first one, we can focus only on ${\rm E}_{\rm el}$, because ${\rm E}_{\rm g}$ is always bounded from below, since the midline is bounded. Therefore, by \eqref{eq:effe_i} we immediately obtained
\begin{equation}
\label{esh}
{\rm E}_{{\rm el}_i}[{\bf z}] \geq C_i \int_{0}^{{L}_i} |{\bf z}_i|^p\, ds + D_iL_i \geq D_iL_i > -\infty.
\end{equation}
Hence, $\Eloop$ is bounded from below and moreover $\Eloop[{\bf z}] = +\infty$ only under complete compression.




Next, consider a sequence $\{ {\bf z}_k\}_{k\in\N}$ such that
$$
\lim_{k} \Eloop[{\bf z}_{k}] = \inf_{{\bf z} \in V} \Eloop[{\bf z}] = m.
$$
Obviously, it exists $\bar{k}$ such that $\forall k \geq \bar{k}$
$$
\Eloop[{\bf z}_k] \leq m+1.
$$
Now, we notice that this sequence is bounded: this follows easily from the boudedness of the clamping parameters and by coercitivity, since
$$
\int_0^{L_i} |{\bf z}_{i_{k}}|^p \, dx \leq \frac{1}{C_i}\int_{0}^{{L}_i} f_i({\bf z}_{i_{k}}(s),s) \, ds - \frac{D_i{L}_i}{C_i}  \leq \frac{1}{C_i}(m+1) - \frac{D_i{L}_i}{C_i} \leq A,
$$
where $A > 0$ is a constant.
Since $V$ is a reflexive space (\cite{brezis} Ch.~4), ${\bf z}_k$ admits a weakly convergent subsequence, {\em i.e.} up to subsequences one has
$$
\exists \, {\bf z} \in V: \quad {\bf z}_k \rightharpoonup {\bf z}.
$$



Now we show that ${\rm E}_{{\rm \rm loop}}[{\bf z}]$ is weakly-lower semicontinuous (WLSC) in $V$. Remember that $\Eloop = {\rm E}_{{\rm el}}+ {\rm E}_{{\rm g}}$ and all linear functionals are LSC, so we can focus on the total stored energy ${\rm E}_{{\rm el}}$. By assumptions made on ${\rm E}_{{\rm el}_{\tiny i}}$, {\em i.e.\ }the hypotheses made on $f_i$, we obtain first that ${\rm E}_{\rm el}$ is WLSC and then the total energy associated to the bounding loop.\\


To introduce and to prove the following theorem, we have to remember that we are looking for the solution of our problem not in a generic Banach space but in $U_{M,n, \eta_i}$, {\em i.e.\ }it has to satisfy the physical and the topological constraints imposed to the problem. Therefore
\begin{thm}
\label{minimol}
If there is at least one admissible 
$$
\overline{{\bf z}} = (\overline{{\bf z}}_1, \overline{{\bf z}}_2) \in \U
$$
with $M \in \R$, $n_i\in \N$ and $\eta_i: [0,L_i] \to \E^3$, then the variational problem described above has a minimizer, i.e. there exists a minimizer ${\bf z} \in \U$ for the loop energy functional. 
\end{thm}
\begin{proof}
Since $\overline{{\bf z}} \in \U$, {\em i.e.\ }it is a competitor, $\U \neq \emptyset$. So, let $\{ {\bf z}_k\}_{k \in \N} \in \U$ a minimizing sequence such that ${\rm E}_{{\rm loop}}[{\bf z}_k]<M$ for some $M \in \R$ be given.\\
By the coercitivity of $f_i$ with $i=1,2$, we obtain that $\U$ is a bounded subset in $V$.
So, we can extract a weakly converging subsequence ${\bf z}_{k_{\tiny h}} \rightharpoonup {\bf z}$. Moreover, as $\U$ is weakly closed in $V$ (Theorem \ref{Umneta}), so ${\bf z} \in \U$.\\
Finally, the weak lower semicontinuity of ${\rm E}_{{\rm loop}}[{{\bf z}}]$ yields
$$
{\rm E}_{{\rm loop}}[{{\bf z}}] \leq \liminf_{h} {\rm E}_{{\rm loop}}[{{\bf z}}_{k_{\tiny h}}] = \lim_{k} {\rm E}_{{\rm loop}}[{{\bf z}}_k] = \inf_{{\bf z} \in U} {\rm E}_{{\rm loop}}[{{\bf z}}],
$$
where ${\bf z}_{k_{\tiny h}}$ is a subsequence of the chosen minimizing sequence ${{\bf z}}_k$. Therefore, the weak limit ${{\bf z}}$ is a global minimizer.
\end{proof}

\subsection{Area-minimizing spanning surface}
Up to now, we only proved the existence of an energy-minimizing configuration for the bounding loop in the absence of the liquid film.
Now we want to show the existence of an area-minimizing spanning surface for the link.

If $\Lambda[{\bf z}]$ is rigid, De Lellis {\em et al.}\ \cite{delellis} proved an important result:
\begin{thm}
\label{area}
Fix ${\bf z} \in V$. If
$$
m_0:= \inf \{ {\rm E}_{{\rm film}}(S): S \in F(\Lambda[{\bf z}], D)\} < +\infty,
$$
then
\begin{enumerate}
\item $F(\Lambda[{\bf z}], D)$ is a {\em good class};
\item there exists $K[{\bf z}]$ a relatively closed subset of $\R^3 \setminus \Lambda[{\bf z}]$ such that $K[{\bf z}] \in F(\Lambda[{\bf z}], D)$ and $K[{\bf z}]$ is a minimizer, i.e. ${\rm E}_{{\rm film}}(K[{\bf z}])=m_0$;
\item $K[{\bf z}]$ is a countably $\hau^2$-rectifiable set and it is an $(\vett{M}, 0, \infty)$-minimal set in $\R^3 \setminus \Lambda[{\bf z}]$ in the sense of Almgren.
\end{enumerate}
\end{thm}

For a precise definition of {\em good class} and $(\vett{M}, 0, \infty)$-minimal set in the sense of Almgren, see respectively \cite{delellis} and \cite{almgren}. In our case the first one is just a family of subsets in which we can control their measures. Namely, it exists a selected and well-defined competitors $L$ with finite $2$-dimensional Hausdorff measure which control the measure of each element of the good class. The second one, instead, is a property of regularity on the subset $K[{\bf z}]$.
The theorem is just a combination of Theorem 2 and 3 in \cite{delellis}.

\subsection{Main result}

Now we come to our main result. Since we are dealing with approximating surfaces, we need to specify the notion of convergence of surfaces. We do this following Fried {\em et al.} \cite{ggg}.

\begin{definizione}
Let $A,B$ be two non empty subsets of a metric space $(M,d_M)$. The {\em Hausdorff distance} between $A$ and $B$ is defined by
$$
d_H(A,B):= \max\{\sup_{a\in A}\inf_{b\in B}d_M(a,b),\,\sup_{b\in B}\inf_{a\in A} d_M(a,b)\}.
$$
\end{definizione}
If we consider all non-empty subsets of $M$, then $d_H$ is a {\em pseudo-metric}, i.e.\ we can always find two subsets $A,B$ with $A \neq B$ such that $d_H(A,B)=0$. However, the set $K(M)$ of non empty {\em compact} subsets of $M$ is a metric space. Moreover, the topology induced by $d_H$ on all closed non empty subsets of $M$ does not depend on $d_M$ and it is said {\em Hausdorff topology}.\\

The problem we have to solve is connected to the fact that $\Lambda[{\bf z}_k]$, the closed subset in $\R^3$ occupied by the whole link, changes along the minimizing sequence. So we have to consider sequences of nonempty closed sets, possibly converging to a closed set, which might be our minimal link. Moreover, since $(K(M), d_{H})$ is not only a metric space, but it is also compact with the distance $d_{H}$, then if we take a bounded sequence in $(K(M), d_{H})$, we can always extract a convergent sequence, using Blaschke's theorem \cite{blaschke}. So, it is reasonable to consider as an assumption of next theorem the existence of a sequence of subsets $\Lambda_k$ which converge to something in the Hausdorff topology, denoted by $\Lambda_k\xrightarrow{H} \Lambda$.


\begin{thm}
\label{primo2}
Let $\Lambda_k$ a sequence of closed non empty subsets of $\E^3$ converging in the Hausdorff topology to a closed set $\Lambda \neq \emptyset$. Assume that
\begin{itemize}
\item[{i}{\rm )}] $\forall k\in\N, S_k\in F(\Lambda_k[{\bf z}], D)$, where $F(\Lambda_k[{\bf z}], D)$ is a good class;
\item[{ii}{\rm )}] $S_k$ is a countably $\hau^2$-rectifiable set;
\item[{iii}{\rm )}] $\hau^2(S_k) = \inf \{\hau^2(S) : S\in F(\Lambda_k[{\bf z}],D)\}<+\infty.$
\end{itemize}
Then the sequence of measures $\mu_k:= \hau^2\lefthalfcup S_k$ is a bounded sequence,  $\mu_k \xrightharpoonup{*} \mu$, up to subsequences, and
$$\mu \geq \hau^2\lefthalfcup S_{\infty},\hbox{ where $S_{\infty}= (\hbox{supt}\, \mu) \setminus \Lambda$ and it is a $\hau^2$-rectifiable set.}$$
\end{thm}


\begin{proof}

Let $\mu_k=\hau^2\lefthalfcup S_k$ and $S_k \in F(\Lambda_k[{\bf z}], D)$. Since $F(\Lambda_k[{\bf z}], D)$ is a good class, for all 
$J\in F(\Lambda_k[{\bf z}],D)$ one has
$$
\mu_k(J)= \hau^2\lefthalfcup S_k(J)=\hau^2(S_k \cap J) \leq \hau^2(J) \leq \hau^2(L)<+\infty,
$$
where $L$ is the generic competitor in the {\em good class}.
Hence, $\mu_k$ is a bounded Radon measure, therefore (\cite{evans}, pp. 54-59), up to a subsequence (not relabeled), $\mu_k\xrightharpoonup{*}\mu$.

Now let $f\in C^{\infty}_c(\R^n)$ with $0 \leq f\leq\chi_{B(x,r)}$; by the weak* convergence of $\mu_k$ we have
$$\int_{B(x,r} f d\mu = \lim_{k} \int_{B(x,r)} f d\mu_k \leq \liminf_{k} \mu_k(B(x,r))$$
so that $\mu(B(x,r)) \leq \liminf_{k} \mu_k(B(x,r))$.

Since $\Lambda_k\xrightarrow H \Lambda$, for any $x\in S_\infty=(\hbox{supt}\,\mu)\setminus \Lambda$ we can find a radius $r$ such that 

$$0<r<d(x,\Lambda)\quad\Longleftrightarrow\quad d(x,\Lambda_k)>\frac12\,d(x,\Lambda).$$

Moreover, if we assume that $\mu(\R^n) =\lim_{k}\mu_k(\R^n)$, combining De Lellis' result with this additional specification, we have that 
$$\mu = \theta\hau^2 \lefthalfcup S_{\infty} \hbox{ with } S_{\infty}=(\hbox{supt}\,\mu) \setminus \Lambda.$$

Indeed
\begin{multline}
\mu(\R^n)= \lim_{k}\mu_k(\R^n) \geq \liminf_{k}\mu_k(\R^n)=\\
\liminf_{k} \mu_k(B(x,r)) + \liminf_{k}\mu_k(\R^n \setminus \overline{B(x,r)})\\
\qquad \qquad \qquad \qquad \qquad \geq \mu(B(x,r)) + \mu(\R^n \setminus \overline{B(x,r)}) = \mu(\R^n),
\end{multline}
where the last equivalence is true because if a set $\{r>0:\mu(\partial B(x,r))\neq 0\}$ is countable, then its Lebesgue measure is zero. Finally, $\mu (\R^n) = \lim_{k}\mu_k(\R^n)$ is obviously true because of the definition of $\mu_k$ and the convergence in the Hausdorff topology.

Let's see that $S_{\infty}$ is a $\hau^2$-rectifiable set. If we fix $x \in S_{\infty}$, i.e. $d(x,\Lambda) > 0$, the function
$r\mapsto\mu (B(x,r))/r^n$ is increasing on $(0, d(x, \Lambda_k))$. By using Preiss' results \cite{preiss}, we can find immediately that
$$\mu = \theta \hau^2 \lefthalfcup \tilde{K}$$
where $\tilde{K}$ is a $\hau^2$-rectifiable set. By the definition of the support of a measure $\tilde{K}=S_{\infty}$. 
\end{proof}

However, this is still not enough. Up to now, we proved in a separate way that the two functionals, the one associated with the elastic link and the other with the film, admit global minimizers. 

Now, first we have to rewrite the second result in terms of the configurations of our system since we only prove the existence of the minimal surface in the presence of a changing boundary, and then we have to write the solution to our problem, i.e.\ making a balance of the two contributions. \\

\begin{thm}
\label{tabular}
Let us suppose that
\begin{itemize}
\item[i)] $\{ {\bf z}_k\}_{k \in \N} \subseteq \U$ a sequence such that ${\bf z}_k \rightharpoonup {\bf z}$ with ${\bf z} \in \U$;
\item[ii)] $S_k \in F(\Lambda[{\bf z}_k], D)$;
\item[iii)] $\gamma$ is a smooth embedding like the one defined in Def \ref{dspanning}.
\end{itemize}
Then there exist two constants $\varepsilon >0$ and $M=M(\varepsilon) >0$ such that $U_{2\varepsilon}(\gamma) \subseteq \E^3 \setminus \Lambda[{\bf z}]$ and $\forall k \geq k_0$
$$
\hau^2(S_k \cap U_{\varepsilon}(\gamma)) \geq M.
$$
\end{thm}
Moreover, Theorem \ref{tabular} says that the intersection between the sequence of the surface $\{S_k\}$, for large $k$, and a neighborhood of the smooth embedding $\gamma$ is not a point but a set with positive measure. Since a tubular neighborhood $U_{\varepsilon}(\gamma)$ exists every time \cite{abate} and that it depends only on the embedding $\gamma$, we can state that the surface $S_{\infty}$ which realizes the area minimal set (Theorem \ref{primo2}) belongs to $F(\Lambda[{\bf z}], D)$. This theorem is fundamental in order to rewrite everything in terms of the configuration ${\bf z}$ only, {\em i.g.\ }we solve the first gap mentioned before.\\
Indeed, suppose true Theorem \ref{primo2} and \ref{tabular}; if we assume by contradiction that
$ S_{\infty} \notin F(\Lambda[{\bf z}], D)$, this would mean that
$$
\exists \gamma \in D_{\Lambda[{\bf z}]}: \left \{
\begin{aligned}
\gamma \cap S_{\infty} &= \emptyset\\
L(\gamma, \vett{{r}}_1) &\neq 1\\
L(\gamma, \vett{{r}}_2) &\neq 1.
\end{aligned}
\right.
$$
If $\gamma \cap S_{\infty} = \emptyset$, we have $\mu(U_{\varepsilon}(\gamma))=0$, with $\varepsilon$ defined in Theorem \ref{tabular}. Hence
$$
0 = \mu(U_{\varepsilon}(\gamma)) \geq \hau^2\lefthalfcup S_{\infty}(U_{\varepsilon}(\gamma)) = \hau^2(S_{\infty} \cap U_{\varepsilon}(\gamma))= \lim_{k} \hau^2( S_k \cap U_{\varepsilon}(\gamma))
$$
which implies
$$
\lim_{k} \hau^2( S_k \cap U_{\varepsilon}(\gamma))=0,
$$
which contradicts the thesis of Theorem \ref{tabular}.
Precisely, both $L(\gamma, \vett{{r}}_1) \neq 1$ and $L_n(\gamma, \vett{{r}}_2) \neq 1$ cannot be achieved because the sequence of $\Lambda_k$ converges in the Hausdorff topology, {\em i.e.} it implies a uniform convergence. Hence, we can state that $F(\Lambda[{\bf z}], D)$ is a weakly closed subset with respect to the weak* convergence.\\

Now, for the proof of Theorem \ref{tabular}, we can say that it is similar to the one presented by Fried et al.\ \cite{ggg} with some modifications. Remember that we are considering a link so, for example, the constant $\eps$ is the same for the whole system and we have to consider the embedding which has the linking number equal to one with both the filaments.
We are now ready to prove our final and main result.
\begin{thm}
Let $M \in \R$, $n_i\in \N$ and $\eta_i:[0,L_i]\to\E^3$ two circumferences be given.
If there exists $\vett{\tilde{{\rm z}}}=(\overline{{\bf z}}_1, \overline{{\bf z}}_2) \in \U$, then there exists a solution ${\bf z}\in \U$ to the Kirchhoff-Plateau problem, i.e. there exists a minimizer ${\bf z}$ for the energy functional ${\rm E}_{KP}$.
\end{thm}
\begin{proof}
Let $\{{\bf z}_k\}$ be a minimizing sequence for ${\rm E}_{KP}$. First of all, by coercivity, we have
$$
\Eloop[{\bf z}_k] \leq C_1 \qquad \hau^2(S_k) \leq C_2,
$$
where $C_1, C_2 >0$ and $S_k \in F(\Lambda[{\bf z}], D)$. Precisely, if ${\bf z}_k \in \U$, by weak closure we can extract a subsequence ${\bf z}_{k_{i}}$ such that
$$
{\bf z}_{k_{i}} \rightharpoonup {\bf z},
$$
where ${\bf z}\in \U$. Now ${\rm E}_{KP}$ is WLSC on $V$. Indeed, by Theorem \ref{minimol}, $\Eloop$ is WLSC, so we only need to show that the functional
\begin{equation}
\label{wlsc}
{\bf z} \mapsto \inf \{ \hau^2(S): S \in F(\Lambda[{\bf z}], D)\} \quad \hbox{is WLSC}.
\end{equation}
To this end, consider $S_k \in F(\Lambda[{\bf z}_k], D)$ such that
$$
\hau^2(S_k) = \inf\{ \hau^2(S): S \in F(\Lambda[{\bf z}_k], D)\} < \infty
$$
By Theorem \ref{primo2}, we find immediately that
$$
\mu \geq \hau^2\lefthalfcup S_{\infty},
$$
where $S_{\infty} = \hbox{supt}\,\mu\setminus \Lambda[{\bf z}]$ and it belongs to $F(\Lambda[{\bf z}], D)$, by the previous remarks. Hence, we obtain the chain of inequalities
\begin{multline}
\liminf_{k} \inf\{ \hau^2(S): S \in F(\Lambda[{\bf z}_k], D)\} \\
\geq \liminf_{k} ( \hau^2(S_k)) = \liminf _{k} \mu_k(\R^3) = \mu(\R^3) \\
\geq \hau^2(S_{\infty}) \geq \inf \{ \hau^2(S): S \in F(\Lambda[{\bf z}], D)\},
\end{multline}
which establishes the lower semicontinuity of the functional \eqref{wlsc} and so the existence of the solution.
\end{proof}

\section{Some simple experiments}

Finally, we tried to get some hint and confirmation reproducing our problem in the laboratory. The film was a solution of 81\% water, 16\% glycerine, 3\% of common dish soap and we added a spoon of baking powder to make it more resistant.

In the first example we took two fixed linked rigid metallic wires in the configuration of fig.~\ref{sistema}, and we observed first a locally minimal configuration consisting of a plane surface and a $D$-spanning set in the sense of our definition. Once the extra surface was removed, the remaining surface seemed to be the minimum surface (fig.~\ref{figure}).

\begin{figure}
\centering
\subfloat[][]
   {\includegraphics[width=.45\textwidth]{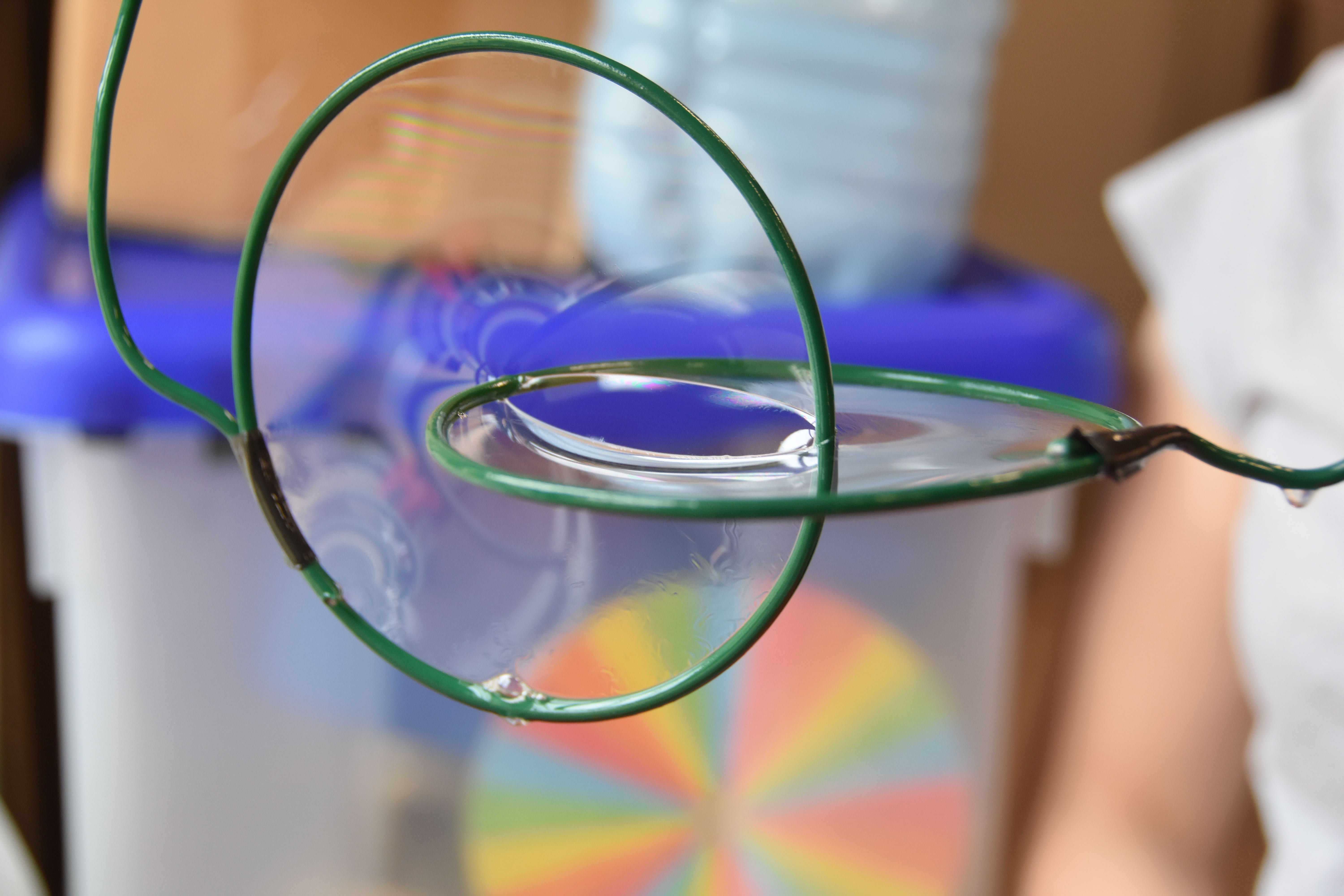}} \quad
\subfloat[][]
   {\includegraphics[width=.45\textwidth]{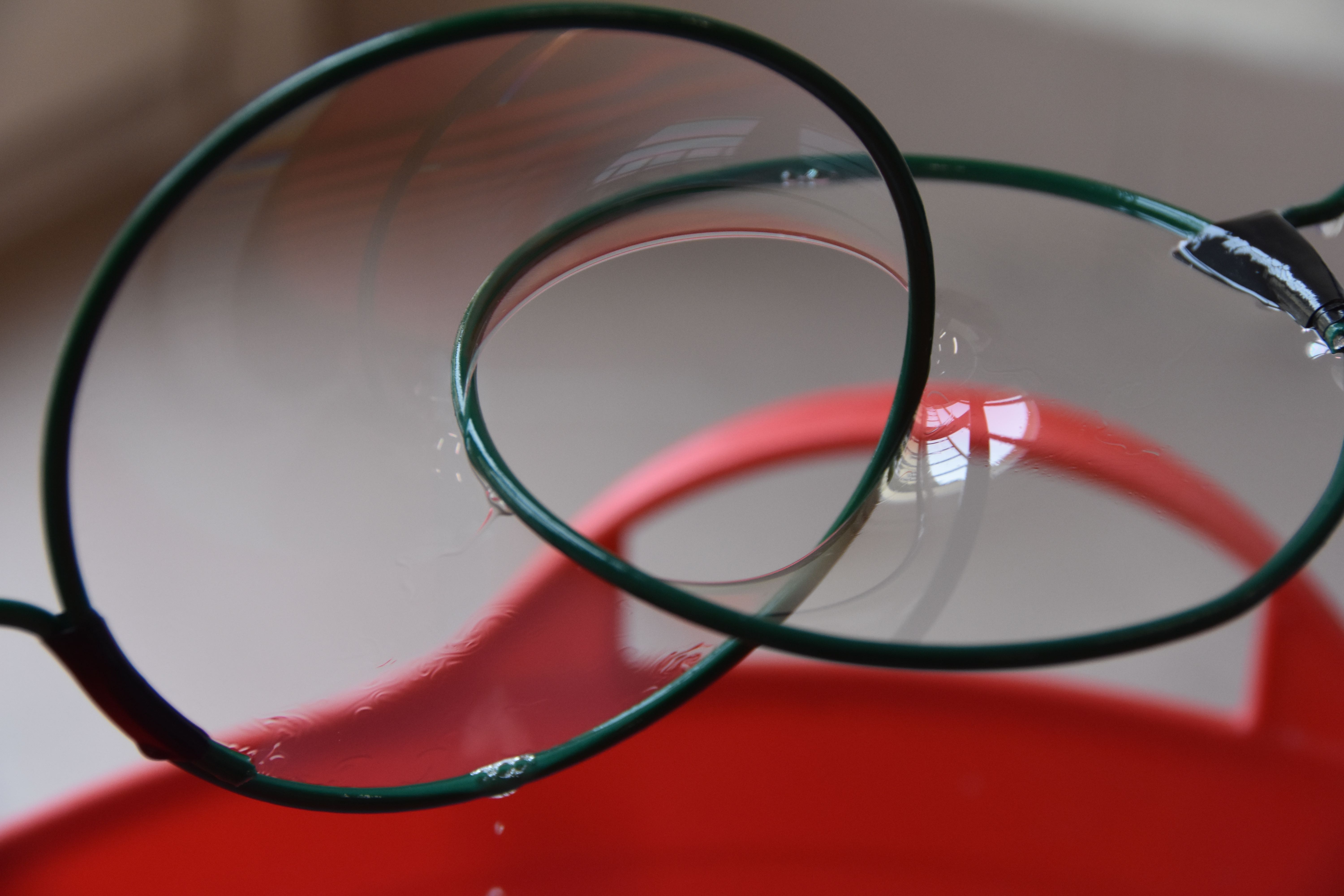}} \\
\caption{Results obtained by our configuration with two fixed linked rigid metallic wires.}
\label{figure}
\end{figure}

In the second experiment, for practical reasons, we took a fixed rigid metallic wire as first rod, while the second was a 0.5 mm thick slender fishing line, twisted of $4$ turns and then glued together.  The pictures (fig.~\ref{sequenza}) show the existence of a configuration balancing the weight of the line with the energy coming from the film.
\begin{figure}
\centering
\subfloat[][]
   {\includegraphics[width=.45\textwidth]{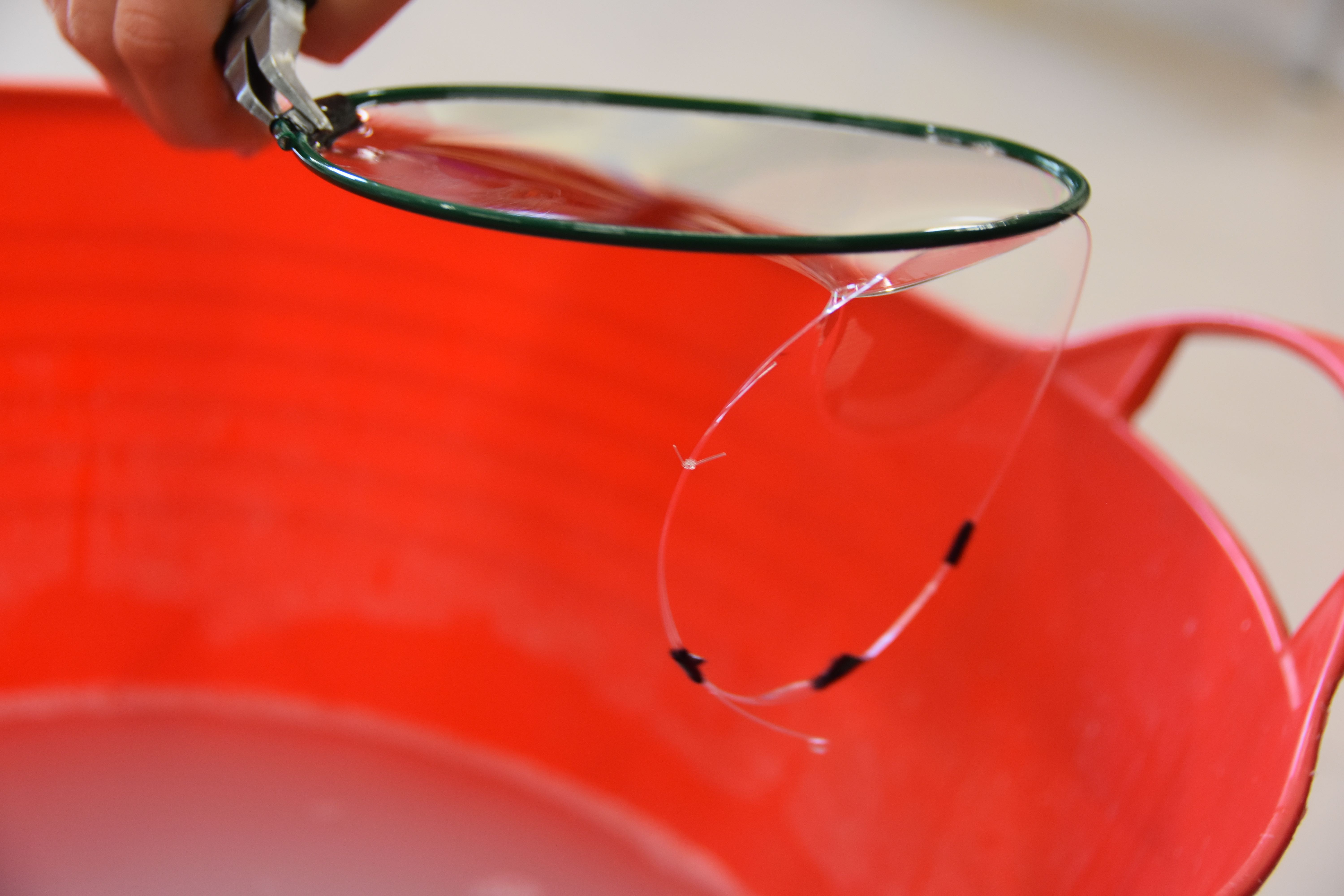}} \quad
\subfloat[][]
   {\includegraphics[width=.45\textwidth]{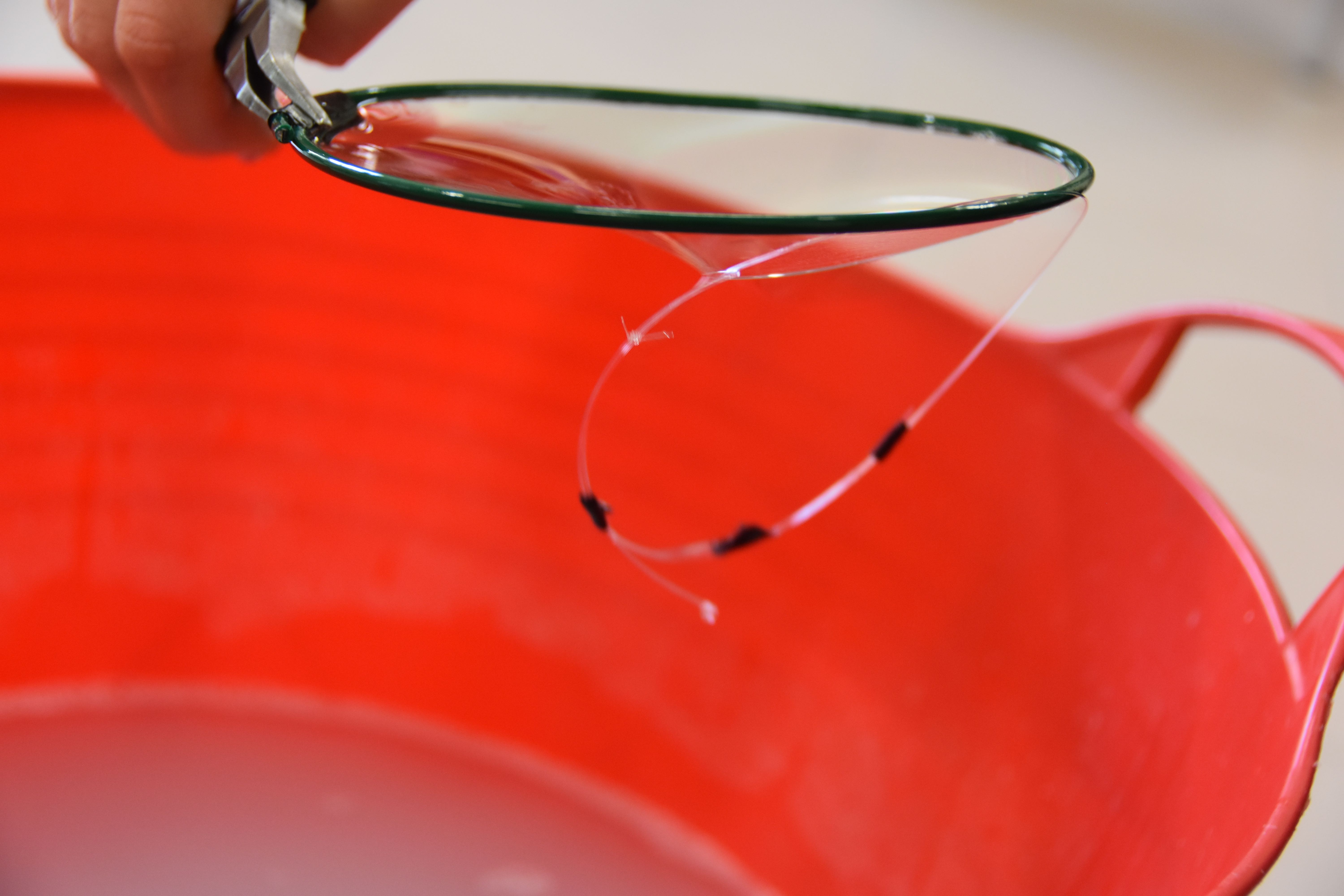}} \\
\caption{Results obtained by our configuration with the mobile component.} 
\label{sequenza}
\end{figure}

\section{Acknowledgements}
The authors wish to thank Marco Degiovanni and Giulio Giuseppe Giusteri for helpful suggestions and fruitful discussions. The work has been partially supported by INdAM–GNFM (Bevilacqua, Marzocchi) and INdAM–GNAMPA (Lussardi).

\end{document}